\documentclass[11pt]{article}

\usepackage{cite}
\usepackage{amsmath,amssymb,amsthm}
\usepackage{algorithm}
\usepackage{algpseudocode}
\usepackage{enumerate}
\usepackage{enumitem}
\usepackage{pgf,tikz}
\usepackage{thmtools}
\usepackage{lscape}
\usepackage{subcaption}
\usepackage{hyperref}
\usepackage{mathrsfs}
\usepackage{enumitem}
\usepackage{multicol}
\usepackage{hyperref}
\usepackage{mathrsfs}
\usepackage[all]{xy}
\usepackage{mathtools}
\usepackage[utf8]{inputenc}
\usepackage[T1]{fontenc}
\usepackage{lmodern}
\usepackage{ulem} 
\usepackage{exscale,relsize}
\usepackage{mathtools}

\usetikzlibrary{arrows,matrix}
\usetikzlibrary{shapes.geometric}

\newtheorem{teo}{Theorem}[section]

\newtheorem{cor}[teo]{Corollary}
\newtheorem{lem}[teo]{Lemma}

\newtheorem{ex}[teo]{Example}

\providecommand{\F}{\ensuremath{\mathbb{F}_{q}} }

\providecommand{\gf}[1]{\ensuremath{\mathbb{F}_{#1}}}

\title{On the dimension of ideals in group algebras, and group codes}

\author{
E. J. García-Claro and H. Tapia-Recillas\\ 
\small{Departmento de Matemáticas}\\
\small{Universidad Autónoma Metropolitana-Iztapalapa, Ciudad de México, México}\\
\small{eliasjaviergarcia@gmail.com}, \small{htr@xanum.uam.mx} }

\begin{document}

\maketitle

\begin{abstract}
Several relations and bounds for the dimension of principal ideals in group algebras are determined by analyzing minimal polynomials of
regular representations. These results are used in the two last sections. First, in the context of semisimple group algebras, to compute, for any abelian code, an element with Hamming weight equal to its dimension. Finally, to get bounds on the minimum distance of certain MDS group codes. A relation between a class of group codes and MDS codes is presented. Examples illustrating the main results are provided.
\end{abstract}


\section{Introduction}

The group algebra $FG$ of a finite group $G$ over the field $F$ is the vector space of formal linear combinations of elements in $G$ with coefficients in $F$, i.e., $FG:=\left\lbrace \sum_{g\in G} a_{g} g \, : \, a_{g}\in F \right\rbrace$. This set is a ring with the usual sum of vectors and the multiplication given by extending  the operation of $G$. If $F$ is finite, a group code (or $G$-code) $C$ is an ideal of $FG$. In addition, if $G$ is abelian, then $C$ is called an abelian group code. The Hamming weight $wt_{G}(\bf{x})$ of an element $\textbf{x} \in FG$ is the number of non-zero coefficients in its coordinate vector with respect to the basis $G$. The minimum weight of a group code is the minimum Hamming weight of its non-zero elements.

Finding ways to compute  the dimension of ideals of finite-dimensional $F$-algebras is itself of interest. In the context of coding theory, this is crucial because the dimension is a parameter  needed, apart from the minimum distance, to determine how good or bad is a code for error correction. As it is well-known, the relations between the length $(n)$, dimension $(k)$, and minimum distance $(d)$ of a linear code are important for various reasons, including the determination of optimal codes or with prescribed minimum distances as in the BCH codes. A relation of this sort is the Singleton bound, $k \leq n-d+1$, which leads to MDS codes when equality holds. Thus, determining the dimension and minimum distance of a linear code, or at least giving upper/lower bounds on these parameters, is an interesting question. Several cases appear in the literature in which these parameters are being explored for group codes. For instance, in \cite{g-equiv} R. A. Ferraz, M. Guerreiro, and C. Polcino determine relations to compute the dimension and minimum weight of abelian codes which are minimal (that is, containing only themselves and the zero ideal) in $\mathbb{F}_{2}(C_{p^n}\times C_{p}) $ where $p$ is an odd prime number and $n\geq 3$. Later, in \cite{raul2},  F. S. Dutra, R. A. Ferraz, and  C. Polcino determine these parameters for two-sided ideals in the semisimple group algebra of a dihedral group. In \cite{elia-gorla}, M. Elia and E. Gorla addressed the problem of determining the dimension of group codes that are principal ideals by studying the characteristic polynomial of the right/left regular representation of a generator. Recently, in \cite{bch-di} lower bounds on these parameters on the class of principal BCH-dihedral codes introduced by these authors are provided and allow to give dihedral codes with prescribed minimum distance.

Principal ideals in finite group algebras are closely related to a class of group codes called Checkable (see \cite{jitman1, jitman2, on-ckcodes}) which are those affording only one check equation, or equivalently, the right/left annihilators of an element. In \cite{on-ckcodes} M. Borello, J. de la Cruz, W. Willems prove that checkable codes are the duals of principal ideals and that the group  algebras for which every group code is checkable are determined, based on the structure of their underlying group. As a consequence of this characterization, a well-known result of Passman is obtained, which states that if $char(F)=p$, all right ideals of $FG$ are principal if and only if $G$ is $p$-nilpotent with a cyclic Sylow $p$-subgroup.


 

In this work, we focus on the determination of relations (such as bounds, identities, and congruences) for the dimension of principal ideals in group algebras by studying the minimal polynomial of the right regular representation determined by a generator of the ideal, and use these relations to study the dimension of semisimple abelian codes. The manuscript is organized as follows. In Section \ref{preliminaries},  preliminary results that will be needed throughout the manuscript are presented. In Section  \ref{app-pri-dec}, by using the Primary Decomposition Theorem, relations for the dimension of some principal ideals in group algebras are presented. These results are first used in Section \ref{abelian} to  study the dimension of abelian codes in semisimple group algebras. In this case, a formula and a bound for the dimension of certain abelian codes are given, and a linear transformation is determined with the property that its evaluation in the generator idempotent of an ideal has Hamming weight equal to the dimension of the ideal. Finally, in Section \ref{mds-gpscod}, these results are used to compute bounds on the minimum distance of a class of MDS group codes. Examples are included illustrating the main results.

\section{Preliminaries}\label{preliminaries}

Throughout this paper, $G$ will denote a finite group, $F$ a field, $R=FG$ the group algebra of $G$ over $F$, and for $b\in R$, $r_b$ ($l_b$) will denote the right (left) regular representation of $b$, i.e., the $F$-endomorphism of $R$ given by $r_{b}(w)=wb$ ($l_{b}(w)=bw$). Also, $m_{b}(x)$ and $p_{b}(x)$ will denote the minimal and characteristic polynomial of $r_{b}$. Furthermore, every module (ideal) is considered a left module (ideal) unless stated otherwise.  

\medskip

Observe that $G$ is an isomorphic group to $\rho(G):=\{r_{g} :\, g\in G\}$ with the composition. Thus    $\varphi:FG\rightarrow F\rho(G)$ given by $\varphi(\sum_{g\in G} a_{g} g)=\sum_{g\in G} a_{g} r_g$ is an isomorphism of $F$-algebras.

\begin{lem}\label{compu-idemp}
 Let $b\in R$, and $\kappa(x)$ be a polynomial that annihilates $r_{b}$. Let $\kappa(x)=f_0(x)f_1(x)$ be a decomposition into coprime factors. If $u_{0}(x), u_{1}(x)\in F[x]$ are such that $u_{0}(x)f_{0}(x) + u_{1}(x)f_{1}(x)=1$, then $u_{i}(b)f_{i}(b)$ is an  idempotent generator of $Rf_{i}(b)$ for $i=0,1$. 
\end{lem}

\begin{proof}
Since $f_{0}(x)$ and $f_{1}(x)$ are coprime, $u_{0}(x), u_{1}(x)\in F[x]$ exist  such that $1=u_{0}(x)f_{0}(x) + u_{1}(x)f_{1}(x)$. Let $E_{i}:=u_{i}(r_{b})f_{i}(r_{b})$ for $i=0,1$. Then $id= E_{0} + E_{1}$, $E_{0}=E_{0}^{2} +E_{0}E_{1}$ and $f_{0}(r_{b})=f_{0}(r_{b})E_{0}+f_{0}(r_{b})E_{1}$, but $E_{0}E_{1}=(u_{0}(r_{b})u_{1}(r_{b}))\kappa(r_b)=0$ and $f_{0}(r_{b})E_{1}=u_{1}(r_{b})\kappa(r_b)=0$, so $E_{0}$ is an idempotent of $F\rho(G)$ and $f_{0}(r_{b})\in F\rho(G)E_{0}$. By a similar argument, $E_{1}$ is an idempotent of $F\rho(G)$ and $f_{1}(r_{b})\in F\rho(G)E_{1}$. The result follows from the fact that $\varphi^{-1}(E_{i})=u_{i}(b)f_{i}(b)\in Rf_{i}(b)$ for all $i$.
\end{proof}

In \cite[Corollary 5]{elia-gorla} a method to compute idempotent generators of a projective
ideal by solving a system of multivariate quadratic and linear equations
over the field is proposed. Lemma \ref{compu-idemp} is an alternative to solve that problem using the Euclidean Algorithm when a special type of generator element is known.

Recall that one of the equivalences of being a projective module is the following \cite[pg 29]{rep3}. $P$ is a projective $A$-module if there exists an $A$-module $P'$ such that  $A^{n}\cong P \oplus P'$ for some $n\in \mathbb{Z}^{+}$. 

\begin{lem}

Let $J$ be a  non-trivial principal ideal of $R$. The following statements are equivalent:

\begin{enumerate}\label{projective}
\item $J$ has a generator $b$ such that $r_b$ is annihilated by a polynomial $\kappa(x)=xh(x)$ where  $x\nmid h(x)$.

\item $J$ is a projective $R$-submodule of $R$. 

\item $J$ has an  idempotent generator.
\end{enumerate}

\end{lem}

\begin{proof}
By \cite[Lemma 2.1, part (a)]{dlc-willems} $2)\Rightarrow 3)$. It is clear that $3)\Rightarrow 1)$. Suppose $J$ has a generator $b$ that is annihilated by a polynomial $\kappa(x)=xh(x)$ where  $x\nmid h(x)$. Then, by Lemma \ref{compu-idemp}, $J$ is generated by an idempotent, so $J$ is projective (by \cite[Lemma 2.1, part (b)]{dlc-willems}).
\end{proof}

Now we will see that the dimensions of the left and the right ideal generated by an element in $R$ are the same. Recall  that the mapping $^*:R \rightarrow R$ given by $u^*=\sum_{g\in G} a_{g}g^{-1}$, for $u=\sum_{g\in G} a_{g}g$, is an antiautomorphism of $F$-algebras (see \cite[Proposition 3.2.11]{grouprings}, \cite[pg 5]{grouprings2}).  

\begin{lem}\label{dim-rank}

 Let $b\in R$, $[r_b]_{G}$ and $[l_b]_{G}$ be the matrices of $r_{b}$ and $l_b$ in the basis $G$, respectively.  Then, $\mathrm{dim}(Rb)=\mathrm{rank}([r_b]_{G})=\mathrm{rank}([l_b]_{G})=\mathrm{dim}(bR)$. 
\end{lem}

\begin{proof}
Since $Rb=\mathrm{span}_{F}( \{gb : \, g\in G\})$, $ \mathrm{dim}(Rb)=\mathrm{dim}(\mathrm{span}_{F}( \{[gb]_G : \, g\in G\})$ where $[gb]_G$ is the  coordinate vector of $gb$ with respect to the basis $G$. But these vectors are precisely the columns of $[r_b]_{G}$, so that $ \mathrm{dim}(Rb)=\mathrm{rank}([r_b]_{G})$. Analogous reasoning with $bR$ and $l_b$, shows that $ \mathrm{dim}(bR)=\mathrm{rank}([l_b]_{G})$. On the other hand,  $(Rb)^*=b^{*}R$ because $(\;)^{*}$ is an antiautomorphism of $F$-algebras. This implies that $\mathrm{rank}([r_b]_{G})=\mathrm{dim}(Rb)=\mathrm{dim}(b^{*}R)=\mathrm{dim}(bR)=\mathrm{rank}([l_{b}]_{G})$. 
\end{proof}

If $F=\gf{q}$, the ideals $Rb$ and $b^{*}R$ define equivalent group codes because $(\;)^{*}$ restricted to $G$ is a permutation.

Lemma \ref{dim-rank} is a different version of \cite[Proposition 1]{elia-gorla}. However, Lemma \ref{dim-rank} mentions the equality of the dimensions between the left and right ideals generated by the same element, while  the mentioned  proposition  does not.


\section{Dimension of ideals in group algebras}\label{app-pri-dec}

For convention, when an integer is matched with a modular class, we mean that the reduction of this number to the respective modulo is equal to the modular class. This notation is the same used in \cite[ Lemma 1.2, Ch.2]{grouprings2}.

\begin{lem} \label{cong}
Let $A\in \mathcal{M}_{n\times n}(F)$ be a matrix with minimal polynomial $m_{A}(x)=x(x-a)^s$ for some integer $s\geq 1$ and $a\in F-\{0\}$. Then the following statements hold:
\begin{enumerate}
\item $\mathrm{trace}(A)a^{-1}$ lies in the prime subfield of $F$.
\item  $\mathrm{rank}(A)=\mathrm{trace}(A)a^{-1}$.
\end{enumerate}
\end{lem}

\begin{proof}

By the Primary Decomposition Theorem \cite[Theorem 12, Ch. 6]{hoffman-kunze}, $n=\mathrm{dim}(ker(A))+\mathrm{dim}(ker(A -aI)^s)$, then $\mathrm{rank}(A)= n-\mathrm{dim}(ker(A))= \mathrm{dim}(ker(A-1)^s)$
. If $N$ is the Jordan canonical form of $A$, then  $\mathrm{trace}(A)=\mathrm{trace}(N)$ is equal to the sum of $a$  as many times as it appears in the main diagonal of $N$. This implies that $\mathrm{trace}(A)a^{-1}$ lies in the prime subfield of $F$. On the other hand, the number of non-zero entries in the main diagonal of $N$ is equal to $\mathrm{dim}(ker(A-aI)^s)$. Thus, if $char(F)=0$,  $\mathrm{rank}(A)=\mathrm{trace}(A)a^{-1}$. If $char(F)=p>0$ and  $u=\mathrm{rank}(A)<p$, then $\mathrm{rank}(A)=\mathrm{trace}(A)a^{-1}$. If $u=pc+r$ where $0\leq r <p$, then $\mathrm{trace}(A)a^{-1}=\overline{r}$ which is equal to the reduction of $u$ modulo $p$, finishing the proof.
\end{proof}

 Lemma \ref{cong}(part 1) is related to \cite[Theorem 7.2.1, 7.2.2]{grouprings}. In fact, if we restrict these last to finite-dimensional group algebras, then Lemma \ref{cong}(part 1) implies both of them.

For any $x\in R$, the coefficient of $x$ at $1$ will be denoted by  $\lambda_{1}(x)$, this is called the trace of $x$ (see, \cite[pg 221]{grouprings}, \cite[pg 31]{grouprings2}).

\begin{teo}\label{dim-primary}
Let $b\in R$ be such that $m_{b}(x)=x^{n}p_{1}(x)^{r_{1}}\cdots p_{t}(x)^{r_{t}} $ where the $p_i$ are monic irreducible distinct factors and $p_{i}\neq x$ for all $i$. Let $\zeta_{n}=\mathrm{dim}(ker(r_{b}^{n})/ ker(r_{b}))$. Then

\begin{enumerate}

\item $ \mathrm{dim}(Rb)= \zeta_{n} + \sum_{i=1}^{t}\mathrm{dim}(ker(p_{i}(r_{b})^{r_{i}}))$. Moreover, if $p_{b}(x)=x^{u}h(x)$ with $x\nmid h(x)$, then $\mathrm{dim}(Rb)= \zeta_{n} + |G|-u$.

\item If $char(F)=p>0$ and $|G|_p$ is the $p$-part of $|G|$, then 

\[\mathrm{dim}(Rb) \geq \begin{cases} (t+1)|G|_{p} &\mbox{if } |G|_{p} \mid \mathrm{dim}(ker(r_{b})) \mbox{                                                and } n>1 \\
t|G|_{p} & \mbox{otherwise}  \end{cases}\]

Moreover, if $n=1$, $|G|_{p} \mid \mathrm{dim}(Rb)$ and $\mathrm{dim}(Rb) \in [t|G|_{p}, |G|-|G|_{p}]$.  

\item  If $m_{b}(x)=x(x-a)^{s}$ for some $s\geq 1$ and $a\in F-\{0\}$, then $\mathrm{dim}(Rb)=|G|\lambda_1(b)a^{-1}$.

\end{enumerate}

\end{teo}

\begin{proof}
Let  $U_{0}=ker(r_{b})$ and $U_{1}=ker(r_{b}^{n})$.
\begin{enumerate}
\item  Let $W\subset U_{1}$ be such that $U_{1}= U_{0} \oplus W$. Then,  by \cite[Theorem 12, Ch. 6]{hoffman-kunze} (part $i$), $R=(U_{0} \oplus W)\oplus ker(p_{1}(r_b)^{r_{1}})\oplus \cdots \oplus ker(p_{t}(r_b)^{r_{t}})$. Thus, by the  rank-nullity theorem,   
\begin{eqnarray*}
 \mathrm{dim}(Rb)&=&\mathrm{dim}(Im(r_{b}))\\
                  &=& \mathrm{dim}(W) + \sum_{i=1}^{t}\mathrm{dim}(ker(p_{i}(r_{b})^{r_{i}}))\\
                &= & \zeta_{n} + \sum_{i=1}^{t}\mathrm{dim}(ker(p_{i}(r_{b})^{r_{i}})). \\
\end{eqnarray*}

Let $p_{b}(x)=x^{u}h(x)$ with $x\nmid h(x)$. Let $U=ker(r_{b}^{|G|})$, $m_{1}(x)$ and $m_{2}(x)$ be the minimal polynomials of $\left.r_{b}\right|_{U_{1}}$ and $\left.r_{b}\right|_{U}$, respectively. By \cite[Theorem 12, part $iii$, Ch. 6]{hoffman-kunze}, $m_{1}(x)=x^{n}$ . Since $U_{1}\subseteq U \subseteq R$ is a chain of $r_{b}$-invariant spaces, then $x^{n}\mid m_{2}(x)\mid m_{b}(x)$. So as $m_{2}(x)\mid x^{|G|}$, $m_{2}(x)\mid x^{n}$, implying that $m_{2}(x)=x^{n}$, and thus $U_{1}=U$. Hence  $\mathrm{dim}(U_{1})=\mathrm{dim}(U)$ which is equal to  the algebraic multiplicity $u$ (see \cite[pp 171-172]{sheldom}). Therefore $ \sum_{i=1}^{t}\mathrm{dim}(ker(p_{i}(r_{b})^{r_{i}}))=|G|-u$.

\item Let $char(F)=p>0$ and $|G|_p$ be the $p$-part of $|G|$. Since $r_{b}$ is a morphism of $R$-modules, by \cite[Theorem 12, part $i$, Ch. 6]{hoffman-kunze}, $R=U_{1}\oplus ker(p_{1}(r_b)^{r_{1}})\oplus \cdots \oplus ker(p_{t}(r_b)^{r_{t}})$ is a decomposition of $R$ as sum of ideals. So $U_{1}$ and $ker(p_{i}(r_{b})^{r_{i}})$ are projective $R$-modules for all $i$, and by \cite[ Corollary 7.16, Ch. VII]{fin-gp}, $|G|_{p}$ divides $ \mathrm{dim}(U_{1})$ and $\mathrm{dim}(ker(p_{i}(r_{b})^{r_{i}}))$ for all $i$. Thus,  $t|G|_{p}\leq \mathrm{dim}(Rb)$. In addition, if $|G|_{p} \mid \mathrm{dim}(U_{0})$, $|G|_{p}$ is a common divisor of $\mathrm{dim}(U_{0})$ and $ \mathrm{dim}(U_{1})$. Thus, if  $n>1$, $ |G|_{p} \mid \zeta_{n} \neq 0$, implying that $(t+1)|G|_{p}\leq \mathrm{dim}(Rb)$.\\
If $n=1$, $Rb$ is projective (by Lemma \ref{projective}), and thus $|G|_{p}$ divides $\mathrm{dim}(Rb)$. So, as $Rb\neq R$,  $\mathrm{dim}(Rb)\in [t|G|_{p}, |G|-|G|_{p}]$.

\item Let $m_{b}(x)=x(x-a)^{s}$ for some $s\geq 1$ and $a\in F-\{0\}$. Then, by Lemmas \ref{dim-rank} and \ref{cong}, $\mathrm{dim}(Rb)=\mathrm{rank}([r_{b}]_G)=\mathrm{trace}([r_{b}]_G)a^{-1}$. Finally, by \cite[Lemma 7.1.1]{grouprings}, $\mathrm{trace}([r_{b}]_G)=|G|\lambda_1(b)$, hence $\mathrm{dim}(Rb)=|G|\lambda_1(b)a^{-1}$.
\end{enumerate}

\end{proof}

 Theorem \ref{dim-primary} (Part $3$)  is a more general version of \cite[Lemma 1.2, part $ ii $, Ch.2] {grouprings2}, which is only  valid  for  idempotents. The benefit of this result when compared with \cite[Lemma 1.2, part $ ii $, Ch.2]{grouprings2} is that it can be applied to a larger amount of elements of $ R $, apart from the idempotents. However,  by Lemma \ref{projective}, any element that satisfies the hypothesis of  Theorem \ref{dim-primary} (part $3$) generates a projective ideal, implying that this result can be applied only to ideals generated by idempotents.


By our convention, if $char(F)=0$ in  Theorem \ref{dim-primary} (part $3$),  we get an explicit formula for the dimension of $Rb$.  However, if $char(F)=p>0$, we only get the class of the dimension modulo $p$ (which is $|G|\lambda_{1}(b)a^{-1}$). Thus we have the following two Corollaries.

\begin{cor}\label{char=0}
Let $b\in R$ and $m_{b}(x)=x(x-a)^{s}$ for some $s\geq 1$ and $a\in F-\{0\}$. If $char(F)=0$, then $\mathrm{dim}(Rb)=|G|\lambda_{1}(b)a^{-1}$.
\end{cor}

\begin{cor}\label{char=p}
Let $b\in R$, $J=Rb$, and $m_{b}(x)=x(x-a)^{s}$ for some $s\geq 1$ and $a\in F-\{0\}$. Let $char(F)=p>0$, and $r$ be the minimum positive integer in the class  $|G|\lambda_{1}(b)a^{-1}$. Then the following holds:

\begin{enumerate}
\item $r \leq \mathrm{dim}(J)$. Moreover, if $\mathrm{dim}(J)\leq p$, then $\mathrm{dim}(J)=r$. In particular, if $|G|-1 \leq p$ and $|G|\neq p$, then the dimension of any non-trivial ideal can be computed in this way. 

\item If $\lambda_{1}(b)=a=1$, $|G|\geq p$ and $c$ is the quotient of dividing $|G|$ by $p$, then $\mathrm{dim}(J)=|G|-pt$ for some $0\leq t\leq c$.

\item  $\mathrm{dim}(J)$ is a multiple of $p$ iff $\lambda_{1}(b)=0$ or $p\mid |G|$. 

\item If  $|G|-1> p$, $c$ is the quotient of dividing $|G|-1$ by $p$, and $\lambda_{1}(b)=0$, then $\mathrm{dim}(J)=pt$ for some $1\leq t \leq c$.

\item If $\mathrm{dim}(J)=1$, then $\lambda_{1}(b)= |G|^{-1}a$. 

\end{enumerate}
\end{cor}

\begin{proof}

\begin{enumerate}

\item As $\mathrm{dim}(J)$ is a positive integer number in the class $|G|\lambda_{1}(b)a^{-1}$, then $r\leq \mathrm{dim}(J)$. Suppose that $\mathrm{dim}(J)\leq p$, then $\mathrm{dim}(J)$  is the minimum positive integer in the class  $|G|\lambda_{1}(b)a^{-1}$, which implies $\mathrm{dim}(J)=r$. If $|G|-1\leq p$ and $|G|\neq p$, $p\nmid |G|$, so that $R$ is semisimple (by \cite[Theorem 3.4.7]{grouprings}). Hence any non-trivial ideal is principal generated by a non-trivial idempotent and has dimension less than or equal to $p$.

\item Suppose that $\lambda_{1}(b)=a=1$. As $J$ is a proper ideal, then $\mathrm{dim}(J)=|G|-pt$  for an integer $1\leq t$, but since the minimum possible value for $\mathrm{dim}(J)$ is $r$ (by part $1$) in such case $t$ would attain the value of $c$.

\item  $\mathrm{dim}(J)= |G|\lambda_{1}(b)a^{-1}=0 $ iff $|G|$ is multiple of $p$ or $\lambda_{1}(b)=0$.

\item Since $\lambda_{1}(b)=0$, $\mathrm{dim}(J)$ is a multiple of $p$, but the greatest multiple of $p$  less than or equal to $|G|-1$ is $pc$, and thus $\mathrm{dim}(J)=pt$ for some $1\leq t \leq c$.

\item If $1=\mathrm{dim}(J)= |G|\lambda_{1}(b) a^{-1}$, then $\lambda_{1}(b)= |G|^{-1}a$.

\end{enumerate}

\end{proof}

In \cite[Theorem 6]{elia-gorla} M. Elia and E. Gorla  give a lower bound on  the dimension of a principal ideal in any group algebra when the multiplicity of $0$ as a root of the characteristic polynomial of the regular left/right representation associated to a generator is known. They also point out that this bound turns out to be the exact dimension when applied to the characteristic polynomial of an idempotent. We note that the only elements for which this equality holds are those with right regular representation having minimal polynomial with $0$ as a simple root. Thus, their result can be restated.

\begin{teo}\cite[Theorem 6]{elia-gorla}\label{reboot}
Let $b\in R$ be such that $p_{b}(x)=x^uh(x)$ where $x\nmid h(x)$,  then
 \[ |G|-u \leq \mathrm{dim}(Rb)\leq |G|-1.\]

Moreover, $\mathrm{dim}(Rb)=|G|-u$ iff $0$ is a simple root of  $m_{b}(x)$.
\end{teo}  

\begin{proof}

By Theorem \ref{dim-primary} (part $1$),  $|G|-u \leq \mathrm{dim}(Rb)$. Furthermore, as $0$ is an eigenvalue of $r_b$, $Rb\subsetneq R$, and so $\mathrm{dim}(Rb)\leq |G|-1$.\\
Let $ m_{b}(x)=x^{n}w(x)$ with $x\nmid w(x)$.
If $ n=1$, by Theorem \ref{dim-primary} (part $1$), $\mathrm{dim}(Rb)=|G|-u$. 
Conversely, if $\mathrm{dim}(Rb)=|G|-u $, $\mathrm{dim}(ker(r_{b}))=u$. As $ker(r_{b})\subseteq ker(r_{b}^{n})\subseteq ker(r_{b}^{|G|})$, then $ker(r_{b})=ker(r_{b}^{n})=ker(r_{b}^{|G|})$. Thus the minimal polynomial of $\left.r_{b}\right|_{ker(r_{b})}=\left.r_{b}\right|_{ker(r_{b}^{n})}$ is $x=x^{n}$, and hence $n=1$.

\end{proof}

\begin{cor}\label{cor-reboot}
Let $b, b' \in R$ be such that $p_{b}(x)=x^{u}h(x)$ and $p_{b'}(x)=x^{u'}h'(x)$. Then the following holds:

\begin{enumerate}
\item If $Rb=Rb'$ and $m_{b}(x)$ has $0$ as a simple root, then $u\leq u'$.

\item If $r_{b}$ is diagonalizable, $\mathrm{dim}(Rb)=|G|-u $. In particular, if $b$ is idempotent this holds.
\end{enumerate}
\end{cor}

\begin{proof}

\begin{enumerate}

\item By Theorem \ref{reboot}, $\mathrm{dim}(Rb)=|G|-u $. Thus $|G|-u' \leq \mathrm{dim}(Rb')\\* = \mathrm{dim}(Rb)=|G|-u $, and so $u\leq u'$.

\item It follows from Theorem \ref{reboot} and  the fact that $r_{b}$ is diagonalizable iff all the roots of $m_{b}(x)$ are simple roots. In particular, if $b$ is idempotent, $m_{b}(x)=x^2-x=x(x-1)$, and hence $r_{b}$ is diagonalizable.
\end{enumerate}
\end{proof}

Every example throughout this work was carried out using SageMath \cite{sage}. 

\begin{ex}\label{ex-dim1} Let $G=\langle u,v \mid u^3=v^2=(uv)^3=1 \rangle=\{1,u,u^2v,v,u^2vu,\\* u^2,vu,uv, uvu, vuv, vu^2, uvu^2\}$ be the alternating group of degree $4$ and $R=\gf{2}G$. If $b=u+u^2vu$ then $m_{b}(x)=x(x^2+x+1)^{2}$. Thus, by Theorem \ref{dim-primary} (part 2), $4\mid \mathrm{dim}(Rb)$ and  $4\leq \mathrm{dim}(Rb)\leq 8$,  so that $\mathrm{dim}(Rb)$ is equal to $4$ or $8$. Alternatively, Theorem \ref{reboot} can be used to compute $\mathrm{dim}(Rb)$. Since $p_{b}(x)=x^{4}(x^2+x+1)^{4}$   and $0$ is a simple root of $m_{b}(x)$, then $\mathrm{dim}(Rb)=12-4=8$. On the other hand, if $b'=1+u+v+u^{2}vu$, then $m_{b'}(x)=x^{2}(x^2+x+1)^{2}$ and  $p_{b'}(x)=x^{4}(x^2+x+1)^{4}$. So, by Theorem \ref{reboot}, $8=12-4 \lneq \mathrm{dim}(Rb)$. In fact, by using Lemma \ref{dim-rank}, we get that $\mathrm{dim}(Rb)=\mathrm{rank}([r_b]_{G})=9$. This happened because the multiplicity of $0$ as a root of $m_{b'}(x)$ is not $1$ but $2$. 
\end{ex}

\begin{ex}\label{ex-dim2}
Let $G= \langle u, v \mid  u^4= 1, \  u^2 = v^2 = (uv)^2, \  vuv^{-1} = u^{-1} \rangle= \{1, u, v, u^2 , u^3v, uv, u^3, u^2v\}$ be the quaternion group and $R=\gf{3}G$. Let $b_{0}=u+2v+2u^2+2u^3v+uv+ u^2v$ and $b_{1}=1+u+v+u^3v$, then $m_{b_{i}}(x)=x (x -1 )^{2}$ for $i=0,1$. Thus, by Corollary \ref{char=p}(part $4$), $\mathrm{dim}(Rb_{0})=3t$ where $1\leq t \leq 2$, i.e., $\mathrm{dim}(Rb_{0})$ is equal to $3$ or $6$.  In addition, by Corollary \ref{char=p}(part $2$), $\mathrm{dim}(Rb_{1})=8-3t$ where $1\leq t \leq 2$, i.e., $\mathrm{dim}(Rb_{1})$ is equal to $5$ or $2$. Let $b_{2}=2+2u+v+u^{3}v$, then $m_{b_{2}}(x)=x(x-2)^{2}$. Thus, by Theorem \ref{dim-primary}(part $3$), $\mathrm{dim}(Rb_{2})=|G|\lambda_{1}(b_{2})2^{-1}=2$ so that $\mathrm{dim}(Rb_{2})$ is equal to $2$ or $5$. In fact, by using Lemma \ref{dim-rank}, we get that $\mathrm{dim}(Rb_{i})$ is equal to $6,5,5$ for $i=0,1,2$, respectively.\\  

Let
 \[ G= \langle u, v \mid    u^2 = v^5 =1, \  uv = v^{-1}u \rangle= \{1, u, v, uv^4, v^2, uv, uv^3, v^4 ,v^3, uv^2\}\]
 be the dihedral group of order $5$. Let  $\alpha$ be a root of the polynomial $p(x)=x^2+2x+2 \in \gf{3}[x]$ in some extension field of $\gf{3}$. $p(x)$ is irreducible, so that $F:=\gf{3}(\alpha )=\gf{9}$. If $R=FG$ and $b=2\alpha ^2+ (\alpha+2)u + (2\alpha +1)v+  \alpha uv^4+
 2v^2+  \alpha uv+  2\alpha ^2uv^3+ 2\alpha ^2v^4+ \alpha v^3+ (\alpha +2)uv^2 \in R$, then $m_{b}(x)=x(x - \alpha ^2)^2$.  Thus, by Theorem \ref{dim-primary} (part $3$), 

\begin{eqnarray*}
\mathrm{dim}_{F}(Rb_{2})&=&|G|\lambda_{1}(b)(\alpha^2)^{-1}\\
                        &=&10(2\alpha ^2)(2\alpha ^2)\\
                        &=&10(4\alpha ^4)\\
                        &=& 2. 
\end{eqnarray*}

Therefore $\mathrm{dim}(Rb_{2})$ is equal to $2$, $5$ or $8$. In fact, by using Lemma \ref{dim-rank}, we get that $\mathrm{dim}(Rb)=8$.  
\end{ex}


\section{Dimension of abelian codes}\label{abelian}

Recall that $F$ is a splitting field for the group  $G$ (the algebra $FG$) if $End_{FG}(V)=F$ for every irreducible $FG$-module $V$ \cite[pg 22]{larry}. Throughout this section, $F=\F$, $p=char(\F)$,  $G$ is abelian of order relatively prime to $q$,  $e$ is a non-trivial idempotent of $R$, and $I=Re$, unless stated otherwise.\\

The following result is a consequence of Theorem \ref{reboot} above.

\begin{cor} 
\label{char-dim}Let $b\in R$ such that $p_{b}=x^{u}h(x)$ with $x\nmid h(x)$. Then $\mathrm{dim}(Rb)=|G|-u$.
 \end{cor}

\begin{proof}

Let $\mathbf{F}$ be a finite extension of $F$ that is a splitting field for $G$ (this exists by\cite[Proposition 7.13]{rep3}). The minimal ideals of $\mathbf{F}G$ have dimension $1$ (by \cite[Corollary 4.4]{larry}), and so their generating idempotents form a basis for $\mathbf{F}G$. Hence the minimal polynomial of $r_{b}$ (seen as a $\mathbf{F}$-automorphism of $\mathbf{F}G$) splits into distinct linear factors. So $0$ is a simple root of $m_{b}(x)\in F[x]$. Thus, by Theorem \ref{reboot}, $\mathrm{dim}(Rb)=|G|-u$.

\end{proof}

Note that Corollary \ref{char-dim} does not depend on the finite condition of $F$.

The mapping $\alpha: G\rightarrow G$ given by $\alpha(x)=x^q$ is an automorphism of $G$, so the group $H:=\langle \alpha \rangle$ acts on $G$ by evaluation. The orbits under this action are called $q$-orbits (also know as $q$-subsets). It is well-known (see, e.g., \cite[Theorem 1.3]{raul}) that  a bijection exists between the minimal ideals of $R$ and the $q$-orbits,  under which the size of a $q$-orbit equals the dimension of the corresponding ideal. We summarize this in the following theorem. 

\begin{teo}\label{barato} 
Let $\{U_{j}\}_{j=1}^{w}$ be the collection of the $q$-orbits of $G$. Let $R=\oplus_{j=1}^{r} I_{j} $ be the decomposition of $R$ into minimal ideals. Then $r=w$ and for a proper indexation, $\mathrm{dim}(I_{j})= |U_{j}|$ for $j=1,...,r$.
\end{teo}

\begin{cor}($q$-orbits bound)\label{cota-q-orbits}
If $Y=\{|U_{j}| : \,  |U_{j}|= |G|\lambda_{1}(e),  \,\, j=1,...,r \}$ and $I$ is a minimal ideal,  then
\[ \min(Y) \leq \mathrm{dim}(I)\leq \max (Y) \]
\end{cor}

\begin{proof}
It follows from Theorems \ref{barato} and \ref{dim-primary} (part $3$).
\end{proof}

By Theorem \ref{barato} $1\leq |Y|$. If $ |Y|=1$, the bound in Corollary \ref{cota-q-orbits}  gives us the exact dimension. \\

The following two results offer a solution to the problem of computing the dimension of any abelian code, but first a set-up is introduced. Let  $m$ be the exponent of $G$. Let $\theta$ be a $m$-th primitive root of unity in some extension field of $F$, then  $\mathbf{F}:=F(\theta)$ is a splitting field for $G$ (see  \cite[Corollary 24.11]{larry}, \cite[Theorem 17.1]{rep3}). Let $\mathbf{R}=\mathbf{F}G$ and $\mathbf{R}=\oplus_{j=1}^{t} \mathbf{R}e_{j} $ be  the decomposition of $\mathbf{R}$ into minimal ideals where $e_{j}$ is idempotent for all $j$. Then $\mathrm{dim}_{\mathbf{F}}(\mathbf{R}e_{j})=1$ for all $j$ (see  \cite[Corollary 4.4]{larry}). This implies that $\eta:=\{e_{j}\}_{j=1}^{t}$ is an $\mathbf{F}$-basis for $\mathbf{R}$. On the other hand,  $\alpha$ can be extended linearly to an $\mathbf{F}$-algebra automorphism of $\mathbf{R}$, and thus $H$ acts on $\eta$ by evaluation (because $\alpha$ sends primitive idempotents into primitive idempotents). Let $U$ be the $\mathbf{F}$-automorphism  of $\mathbf{R}$ that sends $G$ into $\eta$ and $A=[U]_{G}$, then $A[\alpha]_{G}= [\alpha]_{\eta}A$ (see \cite[Theorem 14, Ch. 3]{hoffman-kunze}).    Hence $U$ defines an isomorphism of $H$-sets, and so
\begin{equation}\label{iso-Hsets}
G\cong\eta
\end{equation}
as $H$-sets.

\begin{teo}\label{d-ind}

The following statements hold:

\begin{enumerate}

\item If $f\in R$  is a primitive idempotent, and $D$ is the inverse of $U$, then  $\mathrm{dim}_{F}(Rf)=wt_{G}(D(f))$.

\item  $U$ induces a bijection between the $q$-orbits and the minimal ideals of $R$, under which the size of a $q$-orbit equals the dimension of the corresponding ideal.

\end{enumerate}

\end{teo}

\begin{proof}
Note that $\alpha$ acting in an element of $\eta$ is the same as the inverse of the Frobenius automorphism acting by evaluation on the coefficients of this element. Let $\ast$ and $\odot$ denote the actions of $H$ and $Gal(\mathbf{F}/F)$ in $\eta$, respectively. If $f=\sum_{g \in G} a_{g}g \in \eta $ and $\phi\in Gal(\mathbf{F}/F)$ denotes the Frobenius automorphism, then  $\phi^{-1}\odot f=\phi^{-1}\odot f^{q}= \phi^{-1}\odot ( \sum_{g \in G} \phi(a_{g})g^{q})=\alpha \ast f$. Thus the actions of $H$ and $Gal(\mathbf{F}/F)$ generate the same orbits. 

\begin{enumerate}

\item Let $D$ be the $\mathbf{F}$-automorphism of $\mathbf{R}$ given by $D(x):= U^{-1}(x)$, where $U$ is the change of basis transformation given above. Then  $\left.D\right|_{\eta}:\eta \rightarrow G$ is an isomorphism of $H$-sets. Since $H$ and $Gal(\mathbf{F}/F)$ generate the same orbits in $\eta$, by the Galois descending argument (see  \cite[Proposition $7.18$]{rep3},\cite[Proposition III.6]{split}),  $O\in \eta /H$ exists such that  $f=\sum_{z\in O}z$. Therefore  $D(f)=\sum_{z\in O}D(z)$, which is the sum of  the elements of $G$ belonging to the $q$-orbit $D(O)$, and so $wt_{G}(D(f))=|O|$. On the other hand, by \cite[Theorem III.8]{split},  $\mathrm{dim}_{F}(R f)=\mathrm{dim}_{\mathbf{F}}(\mathbf{R} f)$. Thus, since $\mathbf{R} f=\oplus_{z\in O}\mathbf{R}z$, $\mathrm{dim}_{F}(R f)=|O|=wt_{G}(D(f))$ (because $\mathrm{dim}_{\mathbf{F}}(\mathbf{R}z)=1$ for all $z\in O$).

\item Since $\left.U\right|_{G}: G \rightarrow \eta$ is an isomorphism of $H$-sets,  $U$ induces a size-preserving  bijective correspondence $\widehat{U}:G/H\rightarrow \eta/H$ given by $\widehat{U}(S)=U(S)$. Let $\eta'$ the collection of the primitive idempotents of $R$. Since  $H$ and $Gal(\mathbf{F}/F)$ generate the same orbits in $\eta$, by the Galois descending argument,  $v: \eta/ H \rightarrow \eta'$ given by $v(O)=\sum_{o\in O}o$ is a bijection. Thus $v\circ \widehat{U}$ is a bijection between the $q$-orbits and the primitive idempotents of $R$.  Now, by a similar argument to the one presented at the end of the proof of part $1$,  if $f=v(O)$ where $O \in \eta/ H$, then $\mathrm{dim}_{F}(R f)=|O|=|\widehat{U}^{-1}(O)|=|\widehat{U}^{-1} \circ v^{-1}(f)|$.

\end{enumerate}

\end{proof}

Observe that Theorem \ref{d-ind} (part $2$) implies Theorem \ref{barato}. The $\mathbf{F}$-automor\- -phism $D$ presented in Theorem \ref{d-ind} (part $1$)  will be called \textbf{\textit{ the dimensions indicator of  $R$ associated with $\mathbf{F}$, }} or simply \textbf{\textit{the indicator of }} $R$. Note that  as any abelian code is the direct sum of minimal ideals, the indicator of $R$ can be also applied to compute the dimension of any abelian code. It is shown that the indicator of $R$ is related to the discrete Fourier transform. 
The group of characters $G^{*}$ of $G$ is the set of the group homomorphisms from $G$ to $\mathbf{F}-\{0\}$ with the multiplication of functions. It is well-known that $G^{*}\cong G$ (see, e.g., \cite[Section 1.1]{abelian-c2}). The discrete Fourier transformation $\epsilon$ (see \cite[Section II.A]{split}, \cite[Section 2]{abelian-c2}) is the isomorphism of $\mathbf{F}$-algebras that goes from $\mathbf{F}G^{*}$ to its Artin-Wederburn decomposition $\mathbf{F}^{|G|}$ given by $\epsilon(f)=(f(g))_{g \in G}$. Let $\lambda$ be its inverse, and $\mu$ be the canonical basis of $\mathbf{F}^{|G|}$. Then there is an indexation of  $G^*$ and $\mu$ such that $A=_{G^*}[\lambda]_{\mu}$, and so $A^{-1}=_{\mu}[\epsilon]_{G^*}$ (because $\mathbf{F}G^{*}\cong \mathbf{F}^{|G|}\cong \mathbf{R}$). 
 
Our final step is to provide a way to explicitly compute the indicator of $R$ (i.e., to compute $A^{-1}$) using tensor product algebras and Corollary \ref{char=p}.

Let $G=C_{n_1}\times \cdots \times C_{n_s}$ be a decomposition of $G$ as a product of cyclic groups  with $C_{n_i}=\langle x_{i} \rangle=\{1, x_{i},..., x_{i}^{n_{i}-1}\}$ for all $i$. Let $R_{i}:=\mathbf{F}C_{n_{i}}$,  $l_{i}:R_{i}\rightarrow R_{i}$ be the $\mathbf{F}$-linear transformation given by $l_{i}(y)=x_{i}y$, $\{ \gamma_{ij_{i}}\}_{j_{i}=1}^{n_{i}}$, and $\{ e_{ij_{i}}\}_{j_{i}=1}^{n_{i}}$ be the spectrum of $l_{i}$ and the collection of primitive idempotents of $R_{i}$ for all $i$, respectively.  Let $c_{i}\equiv |C_{n_{i}}|^{-1}\,mod\, p$ for $i=1,...,s$. 

\begin{teo}\label{comp-dime}
Assuming the previous notation, the following holds: $[e_{ij_{i}}]_{C_{n_i}}$ $=c_{i}(1,\gamma_{ij_{i}}^{n_{i}-1}, \gamma_{ij_{i}}^{n_{i}-2}, ..., \gamma_{ij_{i}})$ for $i=1,...,s$ and $j_{i}=1,...,n_{i}$. Furthermore, the coordinate vectors of the primitive idempotents of $\mathbf{R}$ with respect to $G$ are given by $\{[e_{1j_{1}}]_{C_{n_{1}}}\otimes \cdots \otimes [e_{sj_{s}}]_{C_{n_{s}}} : \, j_{i}=1,...,n_{i}  \text{ \, for \,} i=1,...,s  \}$, where $\otimes$ denotes the Kronecker product of vectors.
\end{teo}

\begin{proof}

Let $\beta=C_{n_{1}}\otimes\cdots \otimes C_{n_{s}} $ be the typical basis for the tensor product $R$. As $\mathbf{F}$ is a splitting field for every $C_{n_{i}}$, then every minimal ideal has dimension $1$ as a $\mathbf{F}$-vector space. Thus, since the ideals of $R_{i}$ are $l_{i}$-invariant vector subspaces, every primitive idempotent in $R_{i}$ is an eigenvector of $l_{i}$. Suppose $l_i (e_{ij_{i}})=\gamma_{ij_{i}}e_{ij_{i}}$ where $j_{i}=1,...,n_{i}$ for $i=1,...,s$. The minimal polynomial of $l_{i}$ is $x^{n_{i}}-1$, and therefore $l_{i}$ has as many distinct eigenvalues as $n_{i}$ (because $(|G|,q)=1$), implying that every eigenspace of $l_{i}$ in $R_{i}$ has dimension $1$ for all $i$. Thus, as $1+\gamma_{ij_{i}}^{n_{i}-1}x_i + \gamma_{ij_{i}}^{n_{i}-2}x_{i}^2+ ...+ \gamma_{ij_{i}}x_{i}^{n_{i}-1}$ is an eigenvector of $l_{i}$ associated  to the eigenvalue $\gamma_{ij_{i}}$ for all $i$, then  $e_{ij_{i}}$ must be a multiple of this element, and thus  $e_{ij_{i}}=c_{i}(1+\gamma_{ij_{i}}^{n_{i}-1}x_i + \gamma_{ij_{i}}^{n_{i}-2}x_{i}^2+ ...+ \gamma_{ij_{i}}x_{i}^{n_{i}-1})$ for all $i$ (by Corollary \ref{char=p}, part $5$). 
On the other hand, if  $T=R_{1}\otimes \cdots \otimes R_{s}$, tensor products of the form $e_{1j_{1}}\otimes \cdots \otimes e_{sj_{s}}$ are primitive idempotents of $T$ because the set $\{ e_{1j_{1}}\otimes \cdots \otimes e_{sj_{s}} : \,  j_{i}=1,...,n_{i}  \text{ \, for \,}  i=1,...,s \}$ is a set of orthogonal idempotents with a suitable size, so it must be the set of the primitive idempotents of $T$. In addition, by the definition of $\beta$,  we have that $[e_{1j_{1}}\otimes \cdots \otimes e_{sj_{s}}]_{\beta}= [e_{1j_{1}}]_{C_{n_{1}}}\otimes \cdots \otimes [e_{sj_{s}}]_{C_{n_{s}}}$ for $j_{i}=1,...,n_{i}$ and $i=1,...,s$. Finally, since $\chi:\mathbf{R}\rightarrow T$ given by $\chi(x_{1}^{\epsilon_{1}}\cdots x_{s}^{\epsilon_{s}})=x_{1}^{\epsilon_{1}}\otimes\cdots \otimes x_{s}^{\epsilon_{s}}$ is an isomorphism of $\mathbf{F}$-algebras, the coordinate vectors of the primitive idempotents of $\mathbf{R}$ with respect to $G$ are the same coordinate vectors of the primitive idempotents of $\chi(\mathbf{R})=T$ with respect to $\chi(G)=\beta$.

\end{proof}

Thanks to Theorem \ref{comp-dime}, the indicator of $R$ can be computed as the inverse of the $\mathbf{F}$-linear transformation of $\mathbf{R}$ whose matrix with respect to $G$ has as its columns the coordinate vectors of  the primitive idempotents of $\mathbf{R}$ with respect to $G$. Since $A=_{G^*}[\lambda]_{\mu}$, this  could also have  been achieved using Character theory (see \cite[Corollary $II$.2]{split}), but we were mainly motivated by the fact that the indicator can be obtained as an application of Theorem \ref{dim-primary} (part $3$), with an approach that is independent of the classic one.

\begin{ex}\label{ex-abe}

Let $F=\gf{3}$ and $G=C_2 \times C_4$ where $C_2=\{1,x_{1}\}$ and $C_4=\{1,x_{2},x_{2}^2, x_{2}^3\}$ are the cyclic groups of order $2$ and $4$. Let $\alpha$ be a $4$-th primitive root of the unity whose minimal polynomial over $F$ is $z^2-z-1$. As mentioned before, $\mathbf{F}=F(\alpha)$ is a splitting field for $G$.
 Let $l_{i}$ be as in Theorem \ref{comp-dime}, and $\sigma(l_{i})$ denotes the spectrum of $l_{i}$  for $i=1,2$. Then $\sigma(l_{1})=\{1,2\}$ and $\sigma(l_{2})=\{1,2, \alpha^{2},\alpha^{6} \}$. Thus, by Theorem \ref{comp-dime}, the coordinate vectors of the primitive idempotents of $R_1=\mathbf{F}C_2$ and $R_2=\mathbf{F}C_4$ are $\{ (2,2), (2,1) \}$ and $\{ (1,1,1,1), (1,2,1,2), (1,\alpha^{6}, 2, \alpha^{2}), (1,\alpha^{2}, 2, \alpha^{6})\}$, respectively.
 
Let $\beta$ be as in Theorem \ref{comp-dime}, i.e., $\beta=\{1\otimes 1, 1\otimes x_2, 1\otimes x_{2}^{2}, 1\otimes x_{2}^{3}, x_{1}\otimes 1, x_{1}\otimes x_2, x_{1}\otimes x_{2}^{2}, x_{1}\otimes x_{2}^{3}\}$, then

 \[ \begin{array}{lcl}
(2,2)\otimes(1,1,1,1) & = & (2,2,2,2,2,2,2,2) \\
(2,2)\otimes (1,2,1,2)& = & (2,1,2,1,2,1,2,1) \\
(2,2)\otimes (1,\alpha^{6}, 2, \alpha^{2}) & = & (2,\alpha^{2},1,\alpha^{6},2,\alpha^{2},1,\alpha^{6})\\
(2,2)\otimes (1,\alpha^{2}, 2, \alpha^{6}) & = & (2,\alpha^{6},1,\alpha^{2},2,\alpha^{6},1,\alpha^{2})\\

(2,1)\otimes(1,1,1,1) & = & (2,2,2,2,1,1,1,1) \\
(2,1)\otimes (1,2,1,2)& = & (2,1,2,1,1,2,1,2) \\
(2,1)\otimes (1,\alpha^{6}, 2, \alpha^{2}) & = & (2,\alpha^{2},1,\alpha^{6},1,\alpha^{6},2,\alpha^{2})\\
(2,1)\otimes (1,\alpha^{2}, 2, \alpha^{6}) & = & (2,\alpha^{6},1,\alpha^{2},1,\alpha^{2},2,\alpha^{6})\\
\end{array}\]

 are the coordinate vectors of the primitive idempotents of $R_1 \otimes R_2$ with respect to $\beta$ (because $2\alpha^{6}=\alpha^{2}$). Suppose that $G$ has the ordering determined by $\beta$, i.e., $G=\{1,  x_2,  x_{2}^{2},  x_{2}^{3}, x_{1}, x_{1} x_2, x_{1} x_{2}^{2}, x_{1} x_{2}^{3}\}$, then $_{\beta}[\chi]_{G}=Id$. Hence, these are also the coordinate vectors of the primitive idempotents of $\mathbf{R}$ with respect to $G$, and

\begin{eqnarray*}
 A^{-1}& =& \left(\begin{array}{rrrrrrrr}
2 & 2 & 2 & 2 & 2 & 2 & 2 & 2 \\
2 & 1 & \alpha^{2} & \alpha^{6} & 2 & 1 & \alpha^{2} & \alpha^{6} \\
2 & 2 & 1 & 1 & 2 & 2 & 1 & 1 \\
2 & 1 & \alpha^{6} & \alpha^{2} & 2 & 1 & \alpha^{6} & \alpha^{2} \\
2 & 2 & 2 & 2 & 1 & 1 & 1 & 1 \\
2 & 1 & \alpha^{2} & \alpha^{6} & 1 & 2 & \alpha^{6} & \alpha^{2} \\
2 & 2 & 1 & 1 & 1 & 1 & 2 & 2 \\
2 & 1 & \alpha^{6} & \alpha^{2} & 1 & 2 & \alpha^{2} & \alpha^{6}
\end{array}\right)^{-1}\\
\end{eqnarray*}

\begin{eqnarray*}
  &=&\left(\begin{array}{rrrrrrrr}
1 & 1 & 1 & 1 & 1 & 1 & 1 & 1 \\
1 & 2 & 1 & 2 & 1 & 2 & 1 & 2 \\
1 & \alpha^{2} & 2 & \alpha^{6} & 1 & \alpha^{2} & 2 & \alpha^{6} \\
1 & \alpha^{6} & 2 & \alpha^{2} & 1 & \alpha^{6} & 2 & \alpha^{2} \\
1 & 1 & 1 & 1 & 2 & 2 & 2 & 2 \\
1 & 2 & 1 & 2 & 2 & 1 & 2 & 1 \\
1 & \alpha^{2} & 2 & \alpha^{6} & 2 & \alpha^{6} & 1 & \alpha^{2} \\
1 & \alpha^{6} & 2 & \alpha^{2} & 2 & \alpha^{2} & 1 & \alpha^{6} \\
\end{array}\right).
\end{eqnarray*}

 Let $v_{1}=  20002111$, $v_{2}=  01111101$, $v_{3}=   11212101$, $v_{4}=  21111111$, and $v_{5}=  00200202$. A straightforward computation shows that the element $e_{i}\in R$ such that $[e_{i}]_{G}=v_{i}$ is idempotent for all $i$. We computed $\mathrm{dim}_{F}(Re_{i})$ using Lemma \ref{dim-rank} and this coincided with $wt_{G}(D(e_{i}))$ for all $i$. For instance, $e_{1}$  generates an  $[8,4,4]$-abelian code and $[D(e_{1})]_{G}= 10000111$.  $e_{2}$ generates an  $[8,3,4]$-abelian code and $[D(e_{2})]_{G}= 01000011$. $e_{3}$ generates an  $[8,5,2]$-abelian code and $[D(e_{3})]_{G}= 01111100$. $e_{4}$ generates an  $[8,7,2]$-abelian code and $[D(e_{4})]_{G}=  01111111$. Finally, $e_{5}$  generates an  $[8,6,2]$-abelian code and  $[D(e_{3})]_{G}=01111011$.


\end{ex}

\section{MDS group codes}\label{mds-gpscod}

The Singleton Bound states that if a $[n, k, d]$ linear code over $\gf{q}$  exists, then $k \leq n-d+1$. A code for which the equality is attained in the Singleton Bound is called maximum distance separable, abbreviated MDS. These codes are optimal in the sense that they achieve the maximum possible minimum distance for a given length and dimension, and therefore are of great interest for error correction.  $C$ is said to be a trivial MDS code over $\gf{q}$ if $C = \gf{q}^n$ or $C$ is monomially equivalent to the repetition code  or its dual (see  \cite[pp 71-72]{cod2}). $C$ is an MDS group code if $C$ is an ideal of $\gf{q}G$ such that its parameters satisfy the equality in the Singleton Bound.

In this section, $F=\gf{q}$ and $p=char(F)$ unless stated otherwise. Observe that Corollary \ref{char=p} (part $1$) gives a way to easily compute the dimension of certain group codes, this leads us to our next definition. Let $J$ be an ideal of $R$, if $J$ principal generated by an idempotent, and $\mathrm{dim}_{\gf{q}}(J)\leq p$, then it will be said that $J$ is an \textbf{\textit{easily computable dimension group code}}, abbreviated ECD. If any non-trivial ideal of $R$ is an ECD group code, then it will be said that $R$ is an  \textbf{\textit{easily computable dimension group algebra}} abbreviated ECD. A consequence of Maschke's Theorem (see \cite[Theorem 3.4.7]{grouprings}) is that  $R$ is ECD iff $|G|\leq p+1$ and $|G|\neq p$. The following result is a direct consequence of the Singleton Bound, Theorem \ref{dim-primary} and Corollary \ref{char=p} (part $1$).

\begin{cor}\label{mindis-bounds}
Let $b\in R$ be such that $m_{b}(x)=x^{n}p_{1}(x)^{r_{1}}\cdots p_{t}(x)^{r_{t}} $ where the $p_i$ are monic irreducible distinct factors with $p_{i}\neq x$ for all $i$. Let  $p_{b}(x)=x^uh(x)$ where $x\nmid h(x)$. Let $d$ be the minimum distance of $Rb$. Then the following statements hold:

\begin{enumerate}

\item If $\zeta_{n}=\mathrm{dim}(ker(r_{b}^{n})/ker(r_{b}))$, then $d\leq u-\zeta_{n}+1$. Furthermore, $Rb$ is  an MDS group code iff is an $[|G|, |G|-u+\zeta_{n}, u-\zeta_{n}+1]$-code. In particular, if $n=1$, $d\leq u+1$; $Rb$ is  an MDS group code iff is an $[|G|, |G|-u, u+1]$-code.

\item \[d \leq \begin{cases} |G|-(t+1)|G|_{p} +1 &\mbox{if } |G|_{p} \mid \mathrm{dim}(ker(r_{b})) \mbox{                                                and } n>1 \\
|G|- t|G|_{p} +1 & \mbox{otherwise}  \end{cases}\]

In addition, if $n=1$ and $Rb$ is an MDS group code, then $d \equiv 1 \;(\bmod\; p)$ and $|G|_{p} +1 \leq d \leq |G|- t|G|_{p} +1$.

\item If $m_{b}(x)=x(x-a)^{s}$ for some integer $s\geq 1$ and $a\in F-\{0\}$, $r$ is the minimum positive integer in the class  $|G|\lambda_{1}(b)a^{-1}$; and $Rb$ is an ECD group code, then $d\leq |G|-r+1$. Besides, $Rb$ is  an MDS and ECD  group code iff is an $[|G|, r, |G|-r+1]$-code. 

\end{enumerate}

\end{cor}

\begin{ex} 
Let $G$ be a group of order $mp^{l}$ with $l\geq 1$, $m\neq 1$, and $p\nmid m$. Let $b\in R$, $m_{b}(x)=x^{n}p_{1}(x)^{r_{1}}\cdots p_{t}(x)^{r_{t}} $ where the $p_i$ are monic irreducible distinct factors with $p_{i}\neq x$ for all $i$. Then, by \cite[ Chap. VII, Corollary 7.16]{fin-gp}, $p^{l}$ divides $ \mathrm{dim}(ker(r_{b}^{n}))$ and $\mathrm{dim}(ker(p_{i}(r_{b})^{r_{i}}))$ for  $i=1,...,t$, therefore $1\leq t<m$. So, if  $m=2$, then $m_{b}(x)$ has only two irreducible divisors. Thus, if $n=1$ and $Rb$ is an MDS group code, then $Rb$ is projective (by Lemma \ref{projective}) and its minimum distance $d$ must be $p^{l}+1$ (by Corollary \ref{mindis-bounds}, part $2$).\\
Let $G=\langle a,b \mid a^3=b^2=1, \, bab^{-1}=a^2 \rangle=\{1,b,a,a^2,ba^2,ba\}$ be the symmetric group of degree $3$, and $R=\gf{9}G$. If $\alpha$ is an element of $\gf{9}$ with minimal polynomial $z^2+z+1$, then $b=(2\alpha + 2) + (\alpha + 1)b + \alpha a + (2\alpha + 1)a^2+ (\alpha + 1)ba^2+  ba$ is such that $m_{b}(x)=x(x + 2\alpha)^2$. In this case, $Rb$ is an MDS $[6,3,4]$-code, and so $d=3+1$ as stated in Corollary \ref{mindis-bounds}(part $2$). On the other hand, $b'=(\alpha + 1) + \alpha b + 2 a + 2a^2+  2ba$ is such that $m_{b'}(x)= x^2 (x + \alpha + 2)^2$. In this case, $Rb'$ is an MDS $[6,4,3]$-code, and so $d=3\lneq 3+1$, this happen because the multiplicity of $0$ as a root of $m_{b'}(x)$ is not $1$ but $2$. 
\end{ex}



Now we will study the relation of MDS and ECD group codes. For that purpose we recall the MDS-Conjecture.

\textbf{ MDS-Conjecture\cite[pg 265]{cod2}:} If there is a non-trivial $[n, k]$ MDS code over $\gf{q}$, then     
              
\[n \leq \begin{cases} q + 2 &\mbox{if } q \text{\ even,                                               and } k=3 \mbox{ or } k=q-1\\
q+1 & \mbox{otherwise}  \end{cases}\]

\begin{lem}\label{not-msd}
Let $p$ be a prime number. If the MDS-Conjecture is true, then
the only non-trivial MDS group codes in the non-semisimple group algebra $\gf{p}G$ exist when $G=C_{p}$ and $p$ is odd; and are equivalent to extended Reed-Solomon codes.
\end{lem}

\begin{proof}
If $\gf{p}G$ is non-semisimple, by \cite[Theorem 3.4.7]{grouprings}, $p\mid \, |G|$. Suppose that the MDS-Conjecture is true and that  there exists an MDS $[|G|, k]$ group code $C$ in $\gf{p}G$.  Then if $p=2$, every MDS group code in $\gf{2}G$ is trivial (by \cite[Theorem 2.4.4]{cod1}). If $p$ is an odd prime, then $p\mid \, |G|$ and $|G| \leq p+1$. Since the equality $|G|= p+1$ is not possible, $p\mid \, |G|$ and $|G| < p+1$, thus  $G=C_{p}$. Now, the assertion follows from \cite[Theorem 1]{cyclic-mds}.
\end{proof}

\begin{teo}\label{msd-ecd}
The following statements hold:

\begin{enumerate}
\item If $C$ is an MDS and ECD group code in $R$, then $|G|\leq q+1$

\item Let $p$ be an odd prime. Suppose that the MDS-Conjecture is true and $G\neq C_{p}$. If there exists a non-trivial MDS group code in $\gf{p}G$, then $\gf{p}G$ is an ECD group algebra.
\end{enumerate}
\end{teo}

\begin{proof}

\begin{enumerate}
\item It follows from \cite[Corollary 9.1]{mds}.

\item If there exists a non-trivial MDS group code in $\gf{p}G$, by Lemma \ref{not-msd}, $\gf{p}G$ is semisimple or $G=C_{p}$ with $p$ odd. So  $\gf{p}G$ is semisimple, and by the MDS-Conjecture, $|G|\leq p+1$, implying that  $R$ is an ECD group algebra.
\end{enumerate}
\end{proof}

\section*{Conclusion}
In this paper, we study the dimension of a principal ideal in a group algebra through the minimal polynomial of a generator. This approach allows several new results concerning the dimension of ideals in group algebras which complement some results appearing in the literature. It also provides a way to compute dimensions indicators for abelian codes and presents relations between MDS group codes and easily computable dimension group codes via the MDS-conjecture.
 
\section*{Acknowledgements}

The authors want to thank J. A. Sosaya-Chan for his valuable comments and suggestions. The first author was partially supported by CONACYT (Consejo Nacional de Ciencia y Tecnología, México) under Grant no. 401846.

\end{document}